\documentclass{article}

\usepackage{amsmath}
\usepackage{amsthm}
\usepackage{amssymb}
\usepackage{stmaryrd}
\usepackage{float}

\usepackage{algorithmicx}
\usepackage{algorithm}
\usepackage{algpseudocode}

\newtheorem{thm}{Theorem}
\newtheorem*{thm*}{Theorem}

\newtheorem{prop}{Proposition}

\newtheorem*{op}{Open Problem}

\theoremstyle{definition}
\newtheorem{defn}{Definition}
\newtheorem*{obs}{Observation}

\date{\today}
\author{Jacob Andreas}
\title{The Complexity of Learning Principles and Parameters
Grammars}

\begin{document}
\maketitle

\begin{abstract}
  We investigate models for learning the class of context-free and
  \linebreak context-sensitive languages (CFLs and CSLs).  We begin with a brief
  discussion of some early hardness results which show that unrestricted
  language learning is impossible, and unrestricted CFL learning is
  computationally infeasible;  we then briefly survey the literature on
  algorithms for learning restricted subclasses of the CFLs. Finally, we
  introduce a new family of subclasses, the principled parametric context-free
  grammars (and a corresponding family of principled parametric
  context-sensitive grammars), which roughly model the ``Principles and
  Parameters'' framework in psycholinguistics. We present three hardness
  results: first, that the PPCFGs are not efficiently learnable given
  equivalence and membership oracles, second, that the PPCFGs are not
  efficiently learnable from positive presentations unless $\mathsf{P} =
  \mathsf{NP}$, and third, that the PPCSGs are not efficiently learnable from
  positive presentations unless integer factorization is in {\sf P}.
  \end{abstract}
\newpage

\section{Introduction}

A great deal modern psycholinguistics has concerned itself with resolving the
problem of the so-called ``poverty of the stimulus''---the claim that natural
languages are unlearnable given only the data available to infants, and
consequently that some part of syntax must be ``native'' (i.e. prespecified)
rather than learned. Gold's theorem (described below), which states that there
exists a superfinite class of languages which is not learnable in the limit from
positive presentations, is often offered as proof of this fact (though the
extent to which the theorem is psycholinguistically informative remains a
contentious issue). \cite{gordon90}

But how is innate linguistic knowledge represented? One mechanism usually
offered is the Chomskian ``Principles and Parameters'' framework
\cite{chomsky93}, which suggests that there is a set of universal {\it
principles} of grammar which inhere in the structure of the brain. In this
framework, the process of language learning simply consists of determining
appropriate settings for a finite number of {\it parameters} which determine how
those principles are applied.

While this problem is generally supposed to be easier than unrestricted language
learning, we are not aware of any previous work specifically aimed at studying
the Principles and Parameters model in a computational setting. In this report,
we introduce a family of subclasses of the context-free languages which we
believe roughly captures the intuition behind the Principles and Parameters
model, and explore the difficulty of learning that model in various learning
environments.
 
We begin by presenting an extremely brief survey of the existing literature on
the hardness of language learning; we then introduce three hardness results, one
unconditional, one complexity-theoretic and one cryptographic, which suggest
that the existence of a generalized algorithm for learning in the principles and
parameters framework is highly unlikely. While we obviously cannot produce any
psychologically definitive results in this setting, we at least hope to
challenge the notion that the Principles and Parameters framework is somehow a
computationally satisfying explanation of the language learning process.

\section{Background}

\subsection{Definitions}

\subsubsection{Learnability in the limit}

Gold defines the language learning problem as follows:
\cite{gold}

\begin{defn}
  Given a class of languages $\mathcal{L}$ and an algorithm $A$, we say $A$
  {\bf identifies $\mathcal{L}$ in the limit from positive presentations} if $\forall
  L$, $\forall i_1, i_2, \cdots \in L$, there is a time $t$ such that for all $u
  > t$, $h_u = h_t = A(i_1, i_2, \cdots, i_t)$.
\end{defn}

\subsubsection{Exact identification using queries}
Modeling the language learning process as being entirely dependent on positive
examples seems rather extreme; it's useful to consider enviornments in which the
learner has access to a richer representation of the language.  Angluin
\cite{angluin90} describes a model of language learnability from oracle queries,
as follows:

\begin{defn}
  An {\bf equivalence oracle} for a language $L$ takes as input the
  representation of a language $r(L)$ and outputs ``true'' if $L = L^*$, or some
  $w \in L \Delta L^*$ (the symmetric difference of the languages) otherwise.
  There is an obvious equivalence, first pointed out by Littlestone
  \cite{littlestone88},
  between the equivalence query model and the online mistake bound model.
\end{defn}

\begin{defn}
  A {\bf membership oracle} for a language $L$ with start symbol $S$ takes a
  string $w$, and outputs true if $S \Rightarrow^* w$ and false otherwise.
\end{defn}

\begin{defn}
  A {\bf nonterminal membership oracle} for a language $L$ takes a string $w$ (not
  necessarily in $L$) and a nonterminal $A$, and outputs whether $A \Rightarrow^*
  w$
  (i.e. whether the set of possible derivations with $A$ as a start symbol
  includes $w$).
\end{defn}

\begin{defn}
  A class of languages $\mathcal{L}$ is {\bf learnable from an equivalence
  oracle} (or analogously from an equivalence oracle and a membership oracle,
  sometimes referred to as a ``minimal adequate teacher'') if there exists a
  learning algorithm with runtime polynomial in the size of the representation
  of the class and length of the longest counterexample.
\end{defn}

\subsection{Hardness of language learning}

\begin{thm*}[Gold]
  There exists a class of languages not learnable in the limit from positive
  presentations.
\end{thm*}
\begin{proof}[Proof sketch]
  Construct an infinite sequence of languages $L_1 \subset L_2 \subset \cdots $,
  all finite, and let $L_\infty = \bigcup_i L_i$. Suppose there existed some
  algorithm $A$ that could identify each $L_i$ from positive presentations. Then
  there is a positive presentation of $L_\infty$ that causes $A$ to make an
  infinite number of mistakes. First present a set of examples, all in $L_1$,
  that force $A$ to identify $L_1$. Then present a set of examples forcing it to
  identify $L_2$, then $L_3$, and so on. An infinite number of mistakes can be
  forced in this way, so $L_\infty$ is not learnable in the limit.
\end{proof}

While space does not permit us to discuss the proof here, we also note the
following important result for CFL learning:

\begin{thm*}[Angluin \cite{angluin80}]
    There exists a class of context-free languages with ``natural''
    representations which are not learnable from equivalence queries in time
    polynomial in the size of the representation.
\end{thm*}

\subsection{Learnable subclasses of the CFLs}

While this last result rules out the possibility of a general algorithm for
learning CFLs, subsets of the CFLs have been shown to be learnable when given
slightly more powerful oracles. These include simple deterministic languages
\cite{ishizawa90}, one-counter languages \cite{berman87} and so-called very
simple languages \cite{yokomori91}.
Particularly heartening is Angluin's result that $k$-bounded CFGs can be
learned in polynomial time if nonterminal membership queries are permitted
\cite{angluin87}.

\section{Principled Parametric Grammars}

We now introduce a formal model of the ``principles and parameters'' framework
described in the introduction.

\subsection{Motivation}

Before moving on to the details of the construction, it's useful to consider a
few example ``principles'' and ``parameters'' suggested by proponents of the
model.

\begin{itemize}
  \item 
    {\bf The pro-drop parameter}: does this language allow pronoun dropping? If
    {\tt PNP} is a non-terminal symbol designating a pronoun, this parameter
    determines whether or not a rule of the form $\mathtt{PNP} \rightarrow
    \varepsilon$ exists in the language.
  \item
    {\bf The ergative/nominative parameter}: ergative languages distinguish
    between transitive and intransitive senses of verb by marking the subject,
    while nominative languages (like English) mark the object. Let $\mathtt{NP}$
    and $\mathtt{VP}$ be non-terminal symbols for noun and verb phrases
    respectively, and let $\mathtt{NP_{trans}}$ and $\mathtt{VP_{trans}}$ be
    distinguished versions of those symbols for ergative/nominative marking.
    Now, any language with Verb-Subject-Object order, there will be a rule
    $\mathtt{S} \rightarrow \mathtt{NP}\ \mathtt{VP}$. In an ergative language,
    there is additionally a rule of the form $\mathtt{S} \rightarrow
    \mathtt{NP_{trans}}\ \mathtt{VP}$, and in a nominative language a rule of
    the form $\mathtt{S} \rightarrow \mathtt{NP}\ \mathtt{VP_{trans}}$.
\end{itemize}
In each of these cases, a pattern holds: for every possible possible parameter
setting, there is some finite set of context-free productions in the native
grammar, from which only one must be selected as the element of the learned
grammar. This leads very naturally to the following development of principled
parametric context-free grammars as a model of the principles and parameters
model.

\subsection{Construction}

\begin{defn}
  An {\bf $\mathbf{n}$-principled, $\mathbf{k}$-parametric context-free grammar}
  \linebreak ($(n,k)$-PPCFG) $\Gamma$ is a 4-tuple $(V, \Sigma, \Pi, S)$, where:
  \begin{enumerate}
    \item $V$ is a finite alphabet of nonterminal symbols
    \item $\Sigma$ is a finite alphabet of terminal symbols
    \item $\Pi$ is a set of $n$ production groups of the form
      \[ (A_{i,1} \rightarrow \alpha_{i,1}), \cdots, (A_{i,j} \rightarrow
      \alpha_{i,j}), \cdots, (A_{i,k} \rightarrow \alpha_{i,k}) \]
      where each $\alpha \in (V \cup \Sigma)^*$, i.e. is a finite sequence of
      terminals and nonterminals. Let $\Pi_{i,j}$ denote the production $(A_{i,j}
      \rightarrow \alpha_{i,j})$.
    \item $S \in V$ is the start symbol.
  \end{enumerate}
\end{defn}

\begin{defn}
  A {\bf parameter setting} $p = (p_1, p_2, \cdots, p_n)$ is a sequence of length
  $n$, with each $p_i \in 1..k$. Then define $\Gamma_p$ to be the ordinary
  context-free grammar ($V$, $\Sigma$, $R$, $S$) with $R = \{\Pi_{i,p_i} : i \in
  1..n\}$.
\end{defn}

As usual, let $L(G)$ denote the context free language represented by the CFG
$G$. Then let $\Lambda(\Gamma) = \{ L(G) : \exists p : G = \Gamma_p \}$. 

\begin{defn}
  An algorithm $A$ {\bf learns the PPCFGs from an equivalence oracle} if
  $\forall$ PPCFGs $\Gamma$ and languages $l \in \Gamma$, after a finite number
  of oracle queries, $A$ outputs some $p$ such that $L(\Gamma_p) = l$, or
  determines that no such $p$ exists.
\end{defn}

\begin{defn}
  $A$ {\bf efficiently} learns the PPCFGs from an equivalence oracle if the
  number of oracle queries it makes is bounded by some polynomial function
  $poly(n,k)$.
\end{defn}

\begin{defn}
  Finally, a {\bf principled parametric context-sensitive grammar} is defined
  exactly as above, with corresponding learning definitions, but with
  context-sensitive productions in each production group.
\end{defn}

\subsection{Equivalence}

Some useful facts about the PPCFGs:

\begin{obs}
  A ``heterogeneous PPCFG'' with a variable number of right hand sides can be
  transformed into a ``homogeneous PPCFG'' of the kind described above by
  ``padding'' out the shorter principles with duplicate rules (i.e. to insert an
  unambiguous production $A \rightarrow \alpha$ into an $(n,2)$-PPCFG, add to
  $\Pi$ the production group $(A \rightarrow \alpha), (A \rightarrow \alpha)$).
\end{obs}

\begin{obs}
  A $(n,k)$-PPCFG can be converted into an $(n(k-1),2)$-PPCFG as follows:
  replace each principle
  \[ A \rightarrow (\alpha_1, \alpha_2, \dots, \alpha_k) \]
  with a set of principles
  \begin{align*}
    A_1 &\rightarrow (\alpha_1, A_2) \\
    A_2 &\rightarrow (\alpha_2, A_3) \\
    &\vdots \\
    A_{k-1} &\rightarrow (\alpha_{k-1}, \alpha_k)
  \end{align*}
\end{obs}
Thus without loss of generality we may treat every PPCFG as a $(n,2)$-PPCFG. The
conversion above results in only a polynomial increase in the number of
principles, so any algorithm which is polynomial in $n$, and which assumes
$k=2$, can be used to solve $k > 2$ with only a polynomial increase in running
time. This also means that we may specify an individual language in a PPCFG by a
bit string of length $n$.

Finally, note that a $k$-PPCFG with $n$ rules contains at most $k^n$ languages.

\section{Generic hardness results for PPCFGs}

We will construct a minimal adequate teacher $T$ consisting of two oracles $EQ$
(an equivalence oracle) and $M$ (a membership oracle), such that any algorithm
$A$ requires an exponential number of queries to identify the correct parameter
setting $p$ from a PPCFG $\Gamma$.

\begin{thm}
  Without condition, there exists no algorithm $A$ capable of learning the
  PPCFGs from equivalence queries and membership queries in polynomial time.
\end{thm}

\begin{proof}
  Fix some number $N$. Construct the PPCFG $\Gamma$ with
  \begin{align*}
    V &= X_i : i \in 1..N \\
    S &= \mathtt{START} \\
    \Sigma &= \{ 0, 1 \}
  \end{align*}
  and $\Pi$ as defined as follows:
  \begin{align*}
    &(\mathtt{START} \rightarrow X_1 X_2 \cdots X_N) \\
    &(X_k \rightarrow 0, X_k \rightarrow 1) && \forall k \in 1..N
  \end{align*}

  Every parameter setting $p$ in this grammar allows it to derive precisely 1
  string: every production is deterministic. Consequently, the $N$ possible
  settings of the grammar derive $2^N$ unique strings. Given some algorithm $A$
  for learning PPCFGs, the procedure specified below describes an adversarial
  distinguisher for this PPCFG which forces the learner to make a total of $2^N
  - 1$ queries.

  \begin{algorithm}
    \algblockdefx{Query}{EndQuery}[1][]{{\bf on query} }{{\bf end query}}
    \begin{algorithmic}
      \State $i \gets 0$
      \While{$i < 2^{N}-1$}
        \Query{$EQ(\Gamma')$}
          \If{$\Gamma'$ has not been previously queried}
            \State $i \gets i + 1$
          \EndIf
          \State \Return FALSE, $L(\Gamma')$
          \Comment $L(\Gamma')$ contains only one string
        \EndQuery

        \Query{$M(w)$}
          \If{$w$ has not been previously queried}
            \State $i \gets i + 1$
          \EndIf
          \State \Return FALSE
        \EndQuery
      \EndWhile

      \Query{}
        \State \Return TRUE
        \Comment Only one language is consistent with the evidence
      \EndQuery
    \end{algorithmic}
  \end{algorithm}

  After each query, the number of grammars still possible given the evidence
  provided so far decreases by precisely 1 (because each grammar is capable of
  producing only string), so after $2^N - 1$ queries of either kind, the oracle
  must output true.\\
  
  Thus, only after $2^N - 1$ queries (superpolynomial in $|\Gamma|$ and the
  length of the longest production) can the learner halt, so the grammar is not
  efficiently learnable from membership and equivalence queries.
\end{proof}

\section{Complexity-theoretic hardness results for \\ PPCFGs}

We will construct a reduction from 3SAT to PPCFG learning. Let $X = \{ x_i \}$ be
a set of variables and $C = \{ c_i \}$ be a set of clauses. Let us write $x_j
\in c_i$ if the $j$th $i$th clause is satisfied by the $j$th variable, and
$\bar{x}_j \in c_i$ if the $i$th clause is satisfied by the negation of the
$j$th variable.

Then construct the PPCFG $\Gamma$ with $V = X \cup \{\mathtt{START}\}, \Sigma =
C, S = \mathtt{START}$, and $\Pi$ with the following production groups:
\begin{align*}
  &(\mathtt{START} \rightarrow x_1x_2\cdots x_n) \\
  &(x_i \rightarrow x_{i,T}), (x_1 \rightarrow x_{i,F}) && \forall x_i \in X \\
  &(x_i \rightarrow \varepsilon) && \forall x_i \in X \\
  &(x_{j,T} \rightarrow c_i) && \forall c_i \in C, \forall x_j \in c_i \\
  &(x_{j,F} \rightarrow c_i) && \forall c_i \in C, \forall \bar{x}_j \in c_i
\end{align*}

Note that only for production groups of the form $(x_i \rightarrow x_{i,T}),
(x_1 \rightarrow x_{i,F})$ does the parameter setting change the resulting
language. These groups may be thought of as assigning truth values to the
variables.

\begin{prop}
  \label{l_implies_sat}
  If there exists some $l \in \Gamma$ such that $\forall c_i \in C: c_i \in l$,
  then the 3SAT instance is satisfiable.
\end{prop}

\begin{proof}
  Set $x_i$ true if the rule $x_i \rightarrow x_{i,T}$ is chosen, and false
  otherwise. For any $c_i$ in the language, there is a derivation from
  $\mathtt{START} \Rightarrow^* c_i$ of the following form:
  \begin{align*}
    \mathtt{START} &\Rightarrow x_1 x_2 \cdots x_n \\
    &\Rightarrow x_j \\
    &\Rightarrow x_{j,a} \\
    &\Rightarrow c_i
  \end{align*}
  Then $x_j$ satisfies $c_i$.
\end{proof}

\begin{prop}
  \label{sat_implies_l}
  If the 3SAT instance is satisfiable, there exists some $l \in \Gamma$ such
  that $\forall c_i \in C: c_i \in l$.
\end{prop}

\begin{proof}
  Choose the rule $(x_i \rightarrow x_{i,T})$ if $x_i$ is set true in the
  satisfying assignment, and $(x_i \rightarrow x_{i,F})$ if $x_i$ is set false.
  These settings determine $l$. Then, consider any string $c_i$. There is some
  variable $x_j$ with truth value $a$ which satisfies the corresponding clause;
  then by assignment $l$ contains a production of the form $x_j \rightarrow
  x_{j,a}$, and by definition contains a production of the form $x_{j,a}
  \rightarrow c_i$, so derivation identical to the one in the previous
  proposition must exist.
\end{proof}

\begin{thm}
  \label{complex_ppcfg}
  If $\mathsf{P} \neq \mathsf{NP}$, no efficient algorithm exists for
  learning PPCFGs from positive presentations.
\end{thm}

\begin{proof}
  Assume that there exists some algorithm $A$ which efficiently learns the
  PPCFGs from positive presentations. We will use $A$ to construct a SAT solver
  $S$ by simulating the oracle. Construct $\Gamma$ from the SAT instance as
  described above. Then $S$'s interaction with $A$ takes the following
  form:

  By assumption, after observing polynomially many positive presentations, and
  performing polynomially many computations, $A$ outputs a parameter setting $p$
  which produces every $c_i \in C$, or a signal indicating no such assignment
  exists. From Propositions \ref{l_implies_sat} and \ref{sat_implies_l}, such a $p$
  exists if and only if the SAT instance is satisfiable. Thus $S$ determines in
  a polynomial number of steps whether the SAT instance is satisfiable, and the
  existence of $A$ implies $\mathsf{P} = \mathsf{NP}$.
\end{proof}

\section{Cryptographic hardness results for PPCSGs}

We will construct another reduction, this time from integer factorization to
PPCSG learning. Let $N$ be a product of two $(n-1)$-digit primes.

Let $A$ be a set of non-terminal symbols $A_0..A_{\lceil \lg \sqrt{N} \rceil}$,
and $B, C, Z$ be similar sets of nonterminals of cardinality $\lceil \lg N
\rceil + 1$.
Then construct the PPCSG $\Gamma$ with 
\begin{align*}
  V &= A \cup B \cup C \cup Z \cup \{\mathtt{S}\} \\
  \Sigma &= \{ c_k \} && \forall k \in 0..\lceil \lg N \rceil \\
  S &= \mathtt{S}
\end{align*}
and $\Pi$ with the following production groups:
\begin{align*}
  &(\mathtt{S} \rightarrow A_{\lceil \lg \sqrt{N} \rceil} \mathtt{S}) \\
  &(\mathtt{S} \rightarrow \varepsilon) \\
  &(A_0 \rightarrow B_0), (A_0 \rightarrow \varepsilon) \\
  &(A_j \rightarrow A_{j - 1}),
  (A_j \rightarrow B_j A_{j - 1}) && \forall j \in 1..\lceil
  \lg \sqrt{N} \rceil \\
  &(B_j \rightarrow B_{j-1}B_{j-1}) && \forall j \in 1..\lceil \lg \sqrt{N} \rceil \\
  &(B_0 \rightarrow C_0) \\
  &(C_k C_k \rightarrow C_k Z_k) \\
  &(C_k Z_k \rightarrow C_{k+1} Z_k) && \forall k \in 0..\lceil \lg{N} \rceil
  \\
  &(C_{k+1} Z_k \rightarrow C_{k+1}) \\
  &(C_k \rightarrow c_k)\\
\end{align*}
 
Intuitively, the parameter settings
in this grammar $(A_j \rightarrow A_{j - 1}), (A_j \rightarrow B_j A_{j-1})$
fix some number $m$ between $1$ and $\sqrt{N}$. Each
$A_{\lceil\lg\sqrt{N}\rceil} \Rightarrow^* C^m$, so $S \Rightarrow^* C^{mk}$
for all $k$, i.e. the unary representation of all multiples of $m$. This unary
string may then be collapsed into a $\lceil \lg N \rceil$-ary representation as
a string of terminal $c_i$s.
 
Let $l$ be the language consisting of the single string $w$, where $w$ is the
concatenation of every $a_i$ such that the $i$th digit of the binary
representation of $N$ is 1.

Given a parameter setting $s$ for $\Gamma$, for each production group $(A_j
\rightarrow A_{j - 1}),$ $(A_j \rightarrow B_j A_{j - 1})$ in $s$, let $p_i = 0$
if the first setting is chosen and $1$ if the second setting is chosen. Let
$P_s$ be the number whose binary representation is given by the $p_i$s.
Alternatively, given a binary number $P$ let $s_P$ be the parameter setting
induced by $P$'s bits.

Finally, some notation: given a sequence of strings $S$, let $\bigparallel_{s
\in S} s_i$ denote the concatenation of all $s_i$s.
 
 \begin{prop}
   \label{ppcfg_factoring}
   Given numbers $P$ and $Q$, $P \leq Q$, if $PQ = N$ then $w \in
   L(\Gamma_{s_P})$.
 \end{prop}

\begin{proof}
  In $\Gamma_{s_T}$,
  \begin{align*}
    \mathtt{S} &\Rightarrow^* S^Q \\
    &\Rightarrow^* (A_{\lceil \lg \sqrt{N} \rceil})^Q \\
    &\Rightarrow^* \left( \bigparallel_{\substack{0 \leq i \leq \lceil \lg \sqrt{N}
    \rceil \\
    T_i = 1}} B_i \right)^Q \\
    &\Rightarrow^* B_0^{PQ} = B_0^N \\
    &\Rightarrow^* C_0^N \\
    &\Rightarrow^* w
  \end{align*}
\end{proof}

\begin{prop}
  \label{ppcfg_factoring_2}
  If $w \in L(\Gamma_{s})$, then there exists $Q$ such that $P_s Q=N$.
\end{prop}

\begin{proof}
  Certainly if $w \in L(\Gamma_{s})$, $C_0^{N} \Rightarrow^* w$. But $\mathtt{S}
  \Rightarrow^* C_0^{P_s k}$ for all $k$ (using the derivation in
  Proposition~\ref{ppcfg_factoring}); then there exists some $Q$ such
  that $PQ = N$.
\end{proof}
 
 \begin{thm}
   If integer factorization is hard, no efficient algorithm exists for learning
   random PPCSGs with non-negligible probability from an equivalence oracle.
 \end{thm}
 
 \begin{proof}
   Assume that there exists some algorithm $A$ which, given $\Gamma$ and the
   positive presentation of the single string $w$ as specified above, outputs a
   parameter setting $P$ for $\Gamma$ such that $w \in L(\Gamma_P)$ with
   non-negligible probability a polynomial number of computations. Then we will
   construct a factorizer $F$ that decomposes $N$ into $P$ and $Q$.
   
   From the preceding conjectures, if an acceptable $P$ is found then $PQ = N$,
   for some $Q$, so if $A$ can find a parameter setting in polynomial time then
   this algorithm finds a factorization in polynomial time.
 \end{proof}
 
 This final proof is neither particularly interesting or satisfying: even the
 task of finding a derivation in a CSG is known to be {\sf PSPACE}-complete
 (though it's easy to see that a polynomial-time parsing algorithm for this
 particular family of grammars exists). Note that the only context-sensitive
 production groups employed in this production are used to guarantee a compact
 encoding of $w$; we suspect that there is an alternative way of constructing
 this ``grammar arithmetic'' that requires only weaker rules, perhaps mildly
 context-sensitive or even context-free. We thus close with the following:
 \begin{op}
   If integer factorization is hard, does there exist a polynomial-time
   algorithm for learning random PPCFGs with non-negligible probability from
   positive presentations?
 \end{op}

\section{Conclusion}

We have introduced a new model, the princpled parametric context-free (also
context-sensitive) grammars as a model of the ``Principles and Parameters''
model in psycholinguistics, and presented three hardness-of-learning results for
the class of PPCFGs and PPCSGs. While these results certainly do not demonstrate
definitively that learning under the Principles and Parameters framework is
completely impossible (all that is required for human language learning to be
possible is that one PPCFG be efficiently learnable), we have shown that there
is likely no generic algorithm for learning a class of PPCFGs given either
oracle and membership queries or a positive presentation.  In general, these
results prove that even radically restricting the class of candidate grammars
does not guarantee a successful outcome when attempting to learn CFGs and CSGs.

\end{document}